\documentclass[letterpaper]{article}

\usepackage{natbib,alifeconf}  

\usepackage{booktabs}
\usepackage{amsmath, amssymb} 
\usepackage{amsthm} 
\usepackage{amssymb} 
\usepackage{eqnarray}
\usepackage{url,hyperref,cleveref}
\usepackage{subfig}
\usepackage{mhchem}
\usepackage{xspace}

\newcommand{\Ms}{{\cal M}}
\newcommand{\Rs}{{\cal R}}

\newcommand{\Cr}{{\cal C} \mathord{\uparrow}}

\newcommand{\Ss}{{\bf S}}
\newcommand{\vs}{{\bf v}}
\newcommand{\xs}{{\bf x}}

\newcommand{\prods}{\mathrm{prod}}
\newcommand{\supp}{\mathrm{supp}}
\newcommand{\Ge}{G_{CL}}

\newcommand{\Csr}{{\cal C}{\uparrow}}

\theoremstyle{plain}

\newtheorem{lem}{Lemma}
\theoremstyle{definition}
\newtheorem{defi}{Definition}
\newtheorem{exmp}{Example}

\theoremstyle{remark}

\title{Emergent Complexity in Nuclear Reaction Networks: \\ A Study of Stellar Nucleosynthesis through Chemical Organization Theory}

\author{
    Pedro Maldonado-Lang$^{1,2}$,
    Cl\'ement Vidal$^{3}$
\\
    \mbox{}\\
    $^1$Institute for the Philosophy and Science of Complexity, Chile,
    $^2$Soluciones Analíticas SpA., Chile, pmaldona@sax.cl \\
    $^3$Vrije Universiteit Brussel, Belgium,
    contact@clemvidal.com
} 

\begin{document}

\maketitle
\begin{abstract}
We explore the emergence of complex structures within reaction networks, focusing on nuclear reaction networks relevant to stellar nucleosynthesis. The work presents a theoretical framework rooted in Chemical Organization Theory (COT) to characterize how stable, self-sustaining structures arise from the interactions of basic components. Key theoretical contributions include the formalization of \textit{atom sets} as fundamental reactive units and the concept of \textit{synergy} to describe the emergence of new reactions and species from the interaction of these units. The property of \textit{separability} is defined to distinguish dynamically coupled systems from those that can be decomposed. This framework is then applied to the \mbox{STARLIB} nuclear reaction network database, analyzing how network structure, particularly the formation and properties of atom sets and semi-self-maintaining sets, changes as a function of temperature. Results indicate that increasing temperature generally enhances network cohesion, leading to fewer, larger atom sets. Critical temperatures are identified where significant structural reorganizations occur, such as the merging of distinct clusters of atom sets and the disappearance of small, isolated reactive units. The analysis reveals \textit{core clusters} – large (containing more that 1000 reactions), semi-self-maintaining structures that appear to form the core of all potentially stable nucleosynthetic configurations at various temperatures. Overall, the paper provides insights into the structural underpinnings of stability and emergence in complex reaction networks, with specific implications for understanding stellar evolution and nucleosynthesis.
\end{abstract}

\section*{Introduction: Complex Systems and Reaction Networks}

All life and technology on Earth is based on the manipulation of electromagnetic and chemical forces. However, chemical reactions may be just one way to organize matter into complex, living systems. Until now, nuclear reactions have seemed raw and difficult to control because we only know how to manipulate them for creating power plants or bombs. However, their potential for creating complex, high-energy machines and organizations is completely unexplored. Can we identify reaction networks or assembly routes in nuclear physics that lead to emergent complexity analogous to complex molecular systems in chemistry? What conditions (temperature, density, reaction pathway control) are necessary for guiding nuclear reactions toward more complex products than simple fusion or fission products? To explore these difficult questions, we propose to start with the formalism of reaction networks, as they offer a powerful formalism for modeling complex systems, particularly where interactions involve the transformation of multiple entities, as seen in chemistry, biology, and nuclear astrophysics \citep{Dittrich2007ChemicalOT, Feinberg2019FoundationsOCR}. Stellar formation, for instance, can be viewed as an emergent phenomenon where a dynamically stable system arises from the complex interplay of particles and energy.

This work aims to elucidate mechanisms of organization that allow for the formation and persistence of such stable structures, focusing on nuclear reaction networks within stars. The approach leverages Chemical Organization Theory (COT), an algebraic framework designed for the structural analysis of large reaction networks where direct dynamical simulation is often intractable \citep{Dittrich2007ChemicalOT}.

COT posits that dynamically stable configurations in a reaction network correspond to ``organizations'' – sets of molecular species that are both \textbf{closed} and \textbf{self-maintained} \citep{Dittrich2007ChemicalOT}.
A set of species $X$ is \textit{closed} if all reactions possible with the species in $X$ only produce species that are also in $X$; essentially, no new species or reactions can emerge from within the set. The \textit{closure} of a set $X$, denoted $G_{CL}(X)$, is the smallest closed set containing $X$, formed by iteratively adding products of newly enabled reactions until no more new species are generated.
\textit{Self-maintenance} refers to the ability of a system to continuously regenerate its components through its internal reactions, ensuring that all consumed species are also produced within the set. An organization, being both closed and self-maintained, represents a configuration that can persist over time. Theoretical results link fixed points in the dynamics of a reaction network to such organizations \cite{peter2011relation}, making the structural identification of organizations a crucial step in predicting potential stable states.

This work extends these concepts by introducing finer-grained structural elements and interaction mechanisms, specifically \textit{atom}, \textit{monergy} and \textit{synergy}, to better understand how organizations form and evolve. We present a first study of temperature dependence in stellar nuclear reactions.

More broadly, the study of Artificial Nuclear Life may pave the way towards groundbreaking advances in nuclear physics, materials science, and high-energy astrophysics. It may also provide a theoretical foundation for putative high-energy lifeforms such as stellivores \citep{Vidal2016Stellivore, haqq-misra_projections_2025}.

We start by introducing fundamental concepts for the study of reaction networks in the first~\autoref{sec:Fund_RN}; then describe in the next~\autoref{sec:Close_reac_str} how higher-level structures can be characterized, in particular by formalizing the concepts of monergy and synergy. In the last~\autoref{sec:Stellar_RN}  we apply our framework to the \mbox{STARLIB} database of nuclear reactions. 

\section*{Fundamental Concepts in Reaction Networks} 
\label{sec:Fund_RN}

\subsection{Reaction Networks and Closure}

In order to properly understand the structural and dynamic properties of reaction networks, it is necessary to first grasp the formal foundations of what constitutes such a network. A reaction network consists of a set of species and a collection of interactions among them. These species do not necessarily represent only atoms, as traditionally conceived in chemical reactions; rather, they may also stand for molecules, ecological species, decisions, beliefs, or other abstract entities \citep{heylighen_chemical_2024}. Reactions, in this context, represent the transformation of one set of species into another.

Formally, a reaction network is defined by a finite set of species $\Ms = \{s_1, \dots, s_n\}$ and a finite set of reactions $\mathcal{R} = \{r_1, \dots, r_m\}$, where each reaction $r_i$ is written as:
\[
\sum_{j=1}^{n} a_{ij} s_j \rightarrow \sum_{j=1}^{n} b_{ij} s_j
\]
Here, the coefficients $a_{ij} \ge 0$ and $b_{ij} \ge 0$ are the stoichiometric coefficients corresponding to the supports and products of the reaction $r_i$, respectively. The species $s_j$ for which $a_{ij} > 0$ are the \emph{supports}, and those for which $b_{ij} > 0$ are the \emph{products}. These can be compactly represented by the sets $\supp(r_i) = \{s_j \in \Ms \mid a_{ij} > 0\}$ and $\prod(r_i) = \{s_j \in \Ms \mid b_{ij} > 0\}$, respectively.

\begin{defi}
A reaction network is defined as the pair $\langle \Ms,\Rs \rangle$.
\label{def:RN}
\end{defi}

Given a subset $X \subseteq \Ms$, we define the set of reactions associated to $X$ as:
\[
\Rs_X = \{r \in \Rs \mid \supp(r) \subseteq X \}.
\]

This definition allows us to introduce the notion of when a reaction is said to be \emph{triggerable} by a given set of species.

\begin{defi}
Let $r \in \Rs$; then $r$ is said to be \emph{triggerable} by $X$ if and only if $r \in \Rs_X$.
\label{def:trig}
\end{defi}

Given a set $X$, it is often the case that the products of reactions in $\Rs_X$ are not fully contained within $X$. This leads to the need for a broader notion: if all products of all reactions triggered by $X$ are already in $X$, i.e. $\prods(\Rs_X) \subseteq X$, then $X$ is said to be \emph{closed}.

Note that reactions of the form $\to s$, which in fact correspond to $\emptyset \to s$, where $s \in \Ms$, are always triggered since the empty set $\emptyset$ is present in every set. These reactions are considered inflows to the system, just as reactions of the form $s \to$ are outflows.

\begin{defi}
Let $X \subseteq \Ms$. The set $X$ is said to be \emph{closed} if and only if $\prods(\Rs_X) \subseteq X$.
\label{def:closed}
\end{defi}

If $X$ is not closed, then at least one species produced by a reaction triggered by $X$ lies outside of $X$. Therefore, we must consider an extended set $X \cup \prods(\Rs_X)$, which may still not be closed. This recursive process continues until no new species are generated. Since $\Ms$ is finite, the process necessarily converges to a closed set.

\begin{defi}
Let $X \subseteq \Ms$. The \emph{generated closure} $\Ge(X)$ is defined as the smallest closed set containing $X$.
\label{def:genclos}
\end{defi}

This closure is well defined and unique for each $X \subseteq \Ms$. Consequently, the collection of all closed subsets of $\Ms$ forms a meaningful subset of the power set of $\Ms$, and reflects the generative potential of reactive processes in the network.

\begin{defi}
The collection of all closed subsets $\mathcal{C}$ of a reaction network $\langle \Ms, \Rs \rangle$ is referred to as the \emph{structure of closed sets} of the network.
\label{def:enclosed}
\end{defi}

\subsection{Structural and dynamic stability }

Beyond closure, another essential property in reaction networks is \textit{semi-self-maintenance}. Although closure guarantees that no new species are generated outside a given set, semi-self-maintenance captures the internal balance of consumption and production within that set.

\begin{defi}
Let $X \subseteq \Ms$. The set $X$ is said to be \emph{semi-self-maintaining} if:
\[
\supp(\Rs_X) \subseteq \prods(\Rs_X)
\]
\label{def:semimaint}
\end{defi}

That is, every species consumed by a reaction triggered by $X$ is also produced by at least one of these reactions. This property ensures that there exists a generative structure within $X$ capable of replenishing any resource it consumes, even if this compensation is only qualitative. Therefore, semi-self-maintenance is a necessary structural condition for dynamic persistence, albeit insufficient to guarantee it. For example,consider the reaction set $\{a \rightarrow b,~2b \rightarrow a\}$. The species $\{a, b\}$ form a semi-self-maintaining set, since both consumed species are also produced. However, to the imbalance of producing only one $b$ while needing two to regenerate $a$—this set cannot be considered self-sustaining in a quantitative sense.

To rigorously analyze when quantitative balance is achieved, we introduce the notion of a \emph{process vector} $\vs$, whose entries indicate the frequency or rate at which each reaction occurs over a specified time interval, whether continuous or discrete~\citep{Veloz2017a}.

Let the state of a reaction network be represented by a vector $\xs \in \mathbb{R}^m_+$, where each entry $\xs[j]$ denotes the concentration (or count) of species $s_j$, for $j=1, \dots, m$. The stoichiometry of the network, describing how species are produced and consumed, can be captured using a \emph{stoichiometric matrix} $\mathbf{S} \in \mathbb{Z}^{m \times n}$, defined such that:
\[
\Ss[j,i] = b_{ij} - a_{ij}
\]
where $a_{ij}$ and $b_{ij}$ are the stoichiometric coefficients of the supports and products, respectively, of reaction $r_i$ with respect to species $s_j$.

Given an initial state $\xs$ and a process vector $\vs$, the updated state of the network after applying the reactions at the given rates is given by:
\begin{equation}
\label{st-ch}
\xs_\vs = \xs + \Ss \vs.
\end{equation}

\begin{defi}
A set $X$ is said to be \emph{self-maintaining} if there exists a process vector $\vs$ such that all reactions in $\Rs_X$ are active, i.e., $\vs[i] > 0$ for all $r_i \in \Rs_X$, and the resulting state satisfies $\xs_\vs[j] \geq \xs[j]$ for all $j = 1, \dots, m$.
\end{defi}

A self-maintaining set can thus support its internal dynamics without any net loss in species concentrations, ensuring full quantitative sustainability.

\begin{defi}
Let $X \subseteq \Ms$. Then $X$ is called a (semi-)organization if and only if it is both closed and (semi-)self-maintaining.
\end{defi}

Organizations are sets that neither generate new species (closure) nor suffer depletion (self-maintenance), thereby forming structurally stable configurations. These configurations serve as candidates for dynamically stable regimes in the system and have been shown to correspond to attractors such as fixed points~\citep{seminal}, periodic orbits~\citep{peter2011relation}, and limit cycles~\citep{peter2021linking}.

\section{Structure and Emergence of Reactive Sets}
\label{sec:Close_reac_str}
\subsection{Generating Closed Sets from Smaller Parts}

Dynamically relevant closed sets are those in which all constituent species participate in at least one reaction, either as supports or products. In such sets, every species is subject to the influence of a process vector, thus contributing to the internal dynamics of the system. However, certain species can be dynamically isolated, that is, neither the organization nor the set containing these species is capable of triggering the other. In such cases, the (partial or full) union of these dynamically disconnected species with an existing organization still results in a closed set, yet the added species remain dynamically inert: Their concentrations remain unaffected by the operative process. As a result, these additions are inconsequential for analyzing dynamic behavior within the closed set. Consequently, if the goal is to identify organizations, the problem can be reduced to the analysis of reactive closures alone.

\begin{defi}
Let $X \subset \Ms$. The set $X$ is called reactive if for every $x \in X$ there exists a reaction $r \in \Rs_X$ such that $x \in \prods(r) \cup \supp(r)$.
\end{defi}

\begin{defi}
Let $\Cr$ denote the collection of \textit{closed reactive sets} of the reaction network $\langle \Ms, \Rs \rangle$.
\end{defi}

As noted in \citep{maldonado_lang_2022_16896158,10.1007/978-3-030-19432-1_7_1}, the structures formed by reactive closed sets, as well as by closed sets themselves, satisfy a semi-lattice in which the join operation is given by the closure of the union. This perspective allows the search to be reformulated in terms of reactions rather than species, reducing the number of elements to combine and simplifying the exploration. By applying the union–closure operation to combinations of reactions, one can systematically generate all reactive closures of the network, thereby obtaining the full set of reactive subsets that characterize its structural dynamics.

\subsection{Atoms of the Closure Structure}
If we consider the set of species $X_r=G_{CL}(\text{supp}(r))$ for $r \in \Rs$, we notice that this corresponds to the smallest closed set containing the reaction $r$. We denote this as the {\it atom set} of $r$ \citep{10.1007/978-3-030-19432-1_7_1,maldonado_lang_2022_16896158,peter2023computing}.

Interestingly, we can define an equivalence relation:

\begin{defi}
Let $r_i,r_j \in \Rs$ we define the relation $\mathbf{R}$ as:
$$r_i\mathbf{R}r_j  \Leftrightarrow G_{CL}(\text{supp}(r_i))=G_{CL}(\text{supp}(r_j))$$
\label{def:rela_RN}
\end{defi}
the latter allow us to define an equivalence class:

\begin{defi}

Let $r \in \Rs$ We define the equivalence class of $r$ as $[r] \subset \Rs$ as:
$$ [r] = \{r_i \in \Rs | r_i \mathbf{R} r\} $$
\label{def:equi_class}
\end{defi}
The latter notions can be used to \textit{partition} sets of species. 

\begin{defi}

A \textit{partition} of a set $X$ corresponds to a collection of subsets $\mathbf{X}=\{X_1,X_2,\dots X_n\}$ such that:

\begin{enumerate}
    \item $X=\bigcup_i X_i$
    \item $X_i \cap X_j = \emptyset$ if $ j\neq i$
    \item $X_i \neq \emptyset$ for all $i$
\end{enumerate}
\label{def:partition}
\end{defi}

It is clear that if $x,y \in X$ such that $x\mathbf{R}y$ then $[x]=[y]$, hence equivalence classes of a relation always induce a partition of a set. 
By using the equivalence classes given by the relation \ref{def:rela_RN}, it is possible to define the equivalence classes of reactions:

\begin{defi}
Let $\langle \Ms, \Rs \rangle$ be a network of reactions. The equivalence relation \ref{def:equi_class} generates \textit{a partition} $\mathcal{P}$ such that:

$$ \mathcal{P}=\{[r] | r\in \Rs\} $$

\end{defi}

Since the set of reactions is finite, consequently so will be the partition. Thus, these can be numbered according to their corresponding equivalence class and considered as a collection:

\begin{equation}
    \mathcal{P}=\{P_1,\dots,P_n\}
\end{equation}

With this definition, it is now possible to define the \textit{atom} sets:

\begin{defi}

Let  $\langle \Ms, \Rs \rangle$  be a finite reaction network and $\mathcal{P}=\{P_1,\dots,P_n\}$  be the set of their respective equivalence classes, then we define the \textit{collection of atom sets} $\mathcal{B}$ as:

$$\mathcal{B}= \{B_i=G_{CL}(\supp(P_i)) | P_i \in \mathcal{P}\}$$

\label{def:atom}
\end{defi}

Atom sets arise from disjoint equivalence classes of reactions, yet they may themselves contain other atoms. This nesting leads to the cohesive structures observed in reaction networks (Lemma \ref{lem:monergy}). Accordingly, we define the \textit{level} of an atom as the degree to which it contains other atoms.

\begin{defi}
Let $\langle \Ms, \Rs \rangle$ be a finite reaction network and $B \in \mathcal{B}$ be a atom set. We define the $N_B$ level of $B$ as the number of equivalence classes contained in $B$:

$$N_B=|\{P_i | G_{CL}(\supp(P_i)) \subseteq B\}|$$
\label{def:level}
\end{defi}

To illustrate this, let us consider the following example:

\begin{exmp}

Let $\Rs=\{r_1,r_2,r_3\}$ and $\Ms=\{a,b,c\}$ such that:

\begin{align*}
    r_1: & \quad \emptyset \to a \\
    r_2: & \quad b \to 2b \\
    r_3: & \quad c+b \to 2c 
\end{align*}

Here, it can be noted that the reaction network can be described by three equivalence classes: $P_1=\{r_1\}$, $P_2=\{r_2\}$ and $P_3=\{r_3\}$, which generate the atoms $B_1=\{a\}$, $B_2=\{a,b\}$ and $B_3=\{a,b,c\}$, respectively. Thus, $B_1 \subset B_2 \subset B_3$. Moreover, the respective levels of the atoms correspond to $N_{B_1}=1,N_{B_2}=2$ and $N_{B_3}=3$. 

\end{exmp}

It should be noted, that all the atoms sets $\mathcal{B}$ cover $\Ms$, \textit{i.e.} $\bigcup_i B_i = \Ms, \quad \forall B_i \in \mathcal{B}$. This is because every reaction in $r \in \Rs$ is in at least one atom. Moreover, any reactive closed set $X \in \mathcal{C\!\uparrow}$ can be described as a union of atoms belonging to $\mathcal{B}$ \citep{10.1007/978-3-030-19432-1_7_1,maldonado_lang_2022_16896158}. More precisely: 

\begin{lem}
\label{lem:bas_const}
Let $\mathcal{B}$ be the collection of atoms of the finite reaction lattice $\langle \Ms,\Rs \rangle$.  Then, any $X$ in $\Csr$ is described as a union of atoms:

\begin{equation}
\mathrm{If}\,\, X \in \Csr \quad \mathrm{then} \ \quad X=\bigcup_i B_i \quad \mathrm{such}\,\mathrm{that} \quad B_i \in \mathcal{B}   
\label{eq:closreac}
\end{equation}
\end{lem}

\begin{proof}
Since $X \in \Csr$, it follows that $\mbox{$X=\bigcup G_{CL}(\text{supp}(r))$}$$, \, \forall r \in \Rs_X$. Then, given the definition \ref{def:atom} it is known that $G_{CL}(\text{supp}(r)) \in B$.
\end{proof}

This lemma is particularly significant because it shows that atoms provide a constructive foundation for all reactive closed sets. 

\subsection{Synergy and monergy}
It is important to recall that atoms are not characterized by the particular species they contain, but rather by the equivalence classes they define, i.e. the sets of reactions that generate the same reactive closed sets. These equivalence classes encapsulate the structural roles species play within the reaction network. When two of such equivalence classes (or closed sets) are subjected to the union-closure operation, new reactions may emerge that are absent in the individual components. This phenomenon, known as \textit{synergy}, reveals how novel reactive structures can arise from collective interaction. Notably, synergy is not restricted to closed sets: it can also emerge among non-closed subsets, highlighting the structural potential for new closures induced by latent reaction pathways. This property reflects the dynamic richness of the network, enabling the emergence of behaviors that transcend the simple aggregation of its parts.

To formalize this, we introduce the concept of \textit{synergy} as follows:

\begin{defi}
Let $\mathbf{X} = {X_1, \ldots, X_k}$ be a collection of reactive closed sets, i.e. $X_i \in \Csr$. Then, $\mathbf{X}$ is said to form a \textit{synergy} if and only if:
\begin{enumerate}
\item $\bigcup_{i=1}^k \Rs_{X_i} \subset \Rs_{\mathbf{X}}$
\item For every proper subcollection $\mathbf{X’} \subset \mathbf{X}$, we have $\bigcup_{X \in \mathbf{X’}} \Rs_X = \Rs_{\mathbf{X’}}$
\end{enumerate}
\label{def:sinergy}
\end{defi}

Condition (1) captures the emergence of \emph{new} reactions, those triggered exclusively by the union of all sets in $\mathbf{X}$ and not by any of the subsets individually. Condition (2) ensures that this emergence is genuinely collective: no proper sub-collection of $\mathbf{X}$ independently satisfies the same property, which confirms that synergy depends on the full participation of all components.

Beyond synergy, it is essential to examine a more localized form of emergence. The following result introduces a concept referred to as \textit{monergy} \citep{10.1007/978-3-030-19432-1_7_1,maldonado_lang_2022_16896158}, which describes the emergence of structure from a single reaction. Specifically, if an atom (i.e., a minimal reactive closed set generated from the support of a single reaction) contains one or more nested atoms, and these are all hierarchically ordered by inclusion, then the closure of a single reaction is sufficient to regenerate the entire set. This stands in contrast to synergy, which requires multiple interacting parts.

\begin{lem}
Let $X \in \Cr$ be such that, for every pair of distinct reactive closed sets $X_1 \neq X_2$ contained in $X$, either $X_1 \subset X_2$ or $X_2 \subset X_1$ holds. Then, there exists a reaction $r \in \Rs_X$ such that $X = \Ge(\supp(r))$.
\label{lem:monergy}
\end{lem}

\begin{proof}
Given that all reactive closed sets within $X$ are ordered by strict inclusion, we consider the collection $\mathbf{X} = {X_1, \dots, X_l = X}$ as a chain of nested subsets such that $X_i \subset X_{i+1}$ for all $i \in {1, \dots, l-1}$. Suppose, by contradiction, that there exists no reaction $r \in \Rs$ such that $\Ge(\supp(r)) = X$. Then there must exist a set of reactions $\Rs’ \subset \Rs_X$ whose support species collectively generates $X$ through closure. Let $r_a, r_b \in \Rs’$ be two such reactions that generate $X_a = \Ge(\supp(r_a))$ and $X_b = \Ge(\supp(r_b))$, respectively. If neither $X_a \subset X_b$ nor $X_b \subset X_a$ holds, this contradicts the assumption that all sets in $\mathbf{X}$ are totally ordered by inclusion. Therefore, our assumption must be false and there must exist a single reaction whose support species generates $X$.
\end{proof}

\begin{exmp}
Let us consider the following reaction network:

\begin{align*}
 r1: & \quad a+b \to c \\
 r2: & \quad c+d \to a\\
 r3: & \quad c+2d \to b \\
\end{align*}

Here we can identify two reactive closed sets: 
$X_1=\{a,b,c\}=\Ge(\supp(r_1))$ and 
$X_2=\{a,b,c,d\}=\Ge(\supp(r_2))=\Ge(\supp(r_3))=X$, 
which corresponds to the entire network. Clearly, $X_1 \subset X_2$, and therefore either $r_2$ or $r_3$ triggers the whole network through closure. 

Now, let us consider the same network with a slight modification:

\begin{align*}
 r1: & \quad a+b \to c \\
 r2: & \quad c+d \to a \\
 r3: & \quad e+2d \to b \\
\end{align*}

Consequently, its structure changes. In this case, the closed sets are 
$X_1=\{a,b,c\}=\Ge(\supp(r_1))$, 
$X_2=\{a,c,d\}=\Ge(\supp(r_2))$, 
$X_3=\{b,d,e\}=\Ge(\supp(r_3))$, 
and finally 
$X_4=\{a,b,c,d,e\}=\Ge(\supp(r_2) \cup \supp(r_3))=X$, 
which corresponds to the set containing the entire network. 

In this case, there is no $r \in \Rs_X$ that generates the whole network; rather, two reactions together generate the entire set. Moreover, we have that $X_2 \nsubseteq X_3$ and $X_3 \nsubseteq X_2$.
\label{ex:monergy}
\end{exmp}

This result emphasizes that certain hierarchical structures within the network can be fully induced by a single generative event, which we define as \textit{monergic emergence}. Such structures reflect a compact form of organizational complexity, distinct from the distributed interaction required for synergy.

\section*{Analysis of Stellar Nuclear Reaction Networks}
\label{sec:Stellar_RN}
Given that our objective is to identify \textit{nucleosynthetic units} \citep{burbidge1957synthesis}, our approach focuses on the detection of subsets of nuclear species that are capable of maintaining their internal dynamics, which qualify as \emph{organizations}. To this end, we used the \mbox{STARLIB} nuclear reaction library, which contains approximately 75,000 reactions involving approximately 7,000 distinct nuclear species and particles \citep{sallaska2013starlib}. These reactions span a wide variety of nuclear processes relevant to astrophysical and laboratory contexts, as summarized in table \ref{tab:nuclear-reactions}.

\begin{table*}[t]
    \centering
    \begin{small}
    \begin{tabular}{|c|c|c|}
        \hline
        Type & Reaction & Description and/or example \\ \hline
        (1,1) & $s_1  \to s_2$ & $\beta$-decay or electron/neutron capture: $\ce{^{13}N} \to  \ce{^{13}C} + \ce{e^{+}} + \nu_{\ce{e}}$ \\
        (1,2) & $s_1  \to s_2 + s_3$ & Photodisintegration: $\gamma + \ce{^{28}Si} \to  \ce{^{24}Mg} + \alpha$ \\
        (1,3) & $s_1  \to s_2 + s_3 + s_4$ & Example: $\ce{^{12}C} \to  3\alpha$ (inverse triple-$\alpha$ process) \\
        (2,1) & $s_1 + s_2  \to s_3$ & Proton or neutron capture: $\ce{^{12}C} + p \to  \ce{^{13}N} + \gamma$ \\
        (2,2) & $s_1 + s_2  \to s_3 + s_4$ & Exchange reactions: $\ce{^{15}N} + p \to  \ce{^{12}C} + \ce{^{4}He}$ \\
        (2,3) & $s_1 + s_2  \to s_3 + s_4 + s_5$ & Example: $\ce{^{2}H} + \ce{^{7}Be} \to  p + 2\alpha$ \\
        (2,4) & $s_1 + s_2  \to s_3 + s_4 + s_5 + s_6$ & Example: $\ce{^{3}H} + \ce{^{7}Be} \to  2p + 2\alpha$ \\
        (3,1) & $s_1 + s_2 + s_3  \to s_4$ & Effective three-body collision:  $3\alpha \to \ce{^{12}C}$ \\
        (3,2) & $s_1 + s_2 + s_3  \to s_4 + s_5$ & Effective three-body collision:  $3\alpha \to p + \ce{^{11}B}$ \\
        (4,2) & $s_1 + s_2 + s_3 + s_4  \to s_5 + s_6$ & Effective four-body collision:  $2n + 2\alpha \to \ce{^{7}Li} + \ce{^{3}H}$ \\
        \hline
    \end{tabular}
    \end{small}
    \caption{Classification of nuclear reactions included in \mbox{STARLIB} by the number of supports and products.}
    \label{tab:nuclear-reactions}
\end{table*}

To analyze this extensive dataset, we developed a custom computational library \texttt{pyRN} \cite{pmaldona_pyRN_2023} capable of extracting the structure of reactive closed sets from a given reaction network and determining whether these sets are semi-self-maintaining or constitute full organizations. However, identifying the complete structure of reactive closures is computationally intensive: it is an intrinsically exponential problem that becomes even more demanding when combined with the optimization required to verify self-maintenance. For large-scale networks such as \mbox{STARLIB}, a comprehensive computation of all organizations is therefore infeasible.

To address this challenge, our study adopts a more tractable strategy by focusing exclusively on the identification of \textit{atoms}—that is, the minimal reactive closed sets generated by the closure of individual reactions. We then evaluate whether these atomic sets fulfill the condition of semi-self-maintenance. This targeted approach allows us to significantly reduce computational complexity while still uncovering meaningful candidates for nucleosynthetic units.

It is important to note that this network does not explicitly represent photons as species. Consequently, temperature-dependent processes are modeled as part of an external thermal environment (or thermal bath) in which the reactions are immersed \citep{iliadis2015nuclear}. This design choice allows the network to focus on the intrinsic reaction pathways, while treating thermal effects parametrically.

\subsection*{Temperature Effects on Network Size and atom Sets}

The \mbox{STARLIB} dataset includes activation information for each reaction, enabling us to track which reactions become viable as the temperature increases (typically in Gigakelvin, GK). Notably, the activation of reactions follows an exponential trend with respect to temperature, a behavior that is also reflected by the increasing number of species involved in the activated reactions. This thermally driven expansion of reactive pathways provides a valuable framework for analyzing the progressive emergence of complex nucleosynthetic structures as a function of environmental conditions.

This increase in the number of reactions results in greater network cohesion, an effect clearly observable through the decrease in the number of atoms as temperature rises. However, this metric alone may not adequately capture the structural cohesion of the network. Recall that the level of a closed set (Definition~\ref{def:level}) is defined as the number of atoms contained within it, including the set itself if it constitutes an atom. Consequently, network cohesion is zero when there are no connections among reactions, specifically when no reaction can trigger another. In such an extreme scenario, the network is composed entirely of isolated, disconnected reactions, and thus the number of atoms equals the total number of reactions. At the opposite extreme, cohesion equals one when the network is fully interconnected, with each reaction capable of activating others, resulting in a single atom encompassing the entire network.

As the temperature in the nuclear reaction network increases, the total number of atoms decreases, while the remaining atoms tend to be larger. This reduction in the number of atoms is caused by the emergence of new reactions activated by increasing temperature. The appearance of these additional reactions facilitates greater interconnectivity, thus enhancing structural cohesion.

To formally quantify this cohesive property, we define the cohesion $h$ of a network as follows:

\begin{defi}
The cohesion $h$ of a network is defined as one minus the ratio of the number of atoms $|\mathcal{B}|$ to the total number of reactions $|\Rs|$, that is:
$$h = 1 - \frac{|\mathcal{B}|}{|\Rs|}.$$
\label{def:cohesion}
\end{defi}

\begin{figure}[h]
    \centering
    \includegraphics[width=0.4\textwidth]{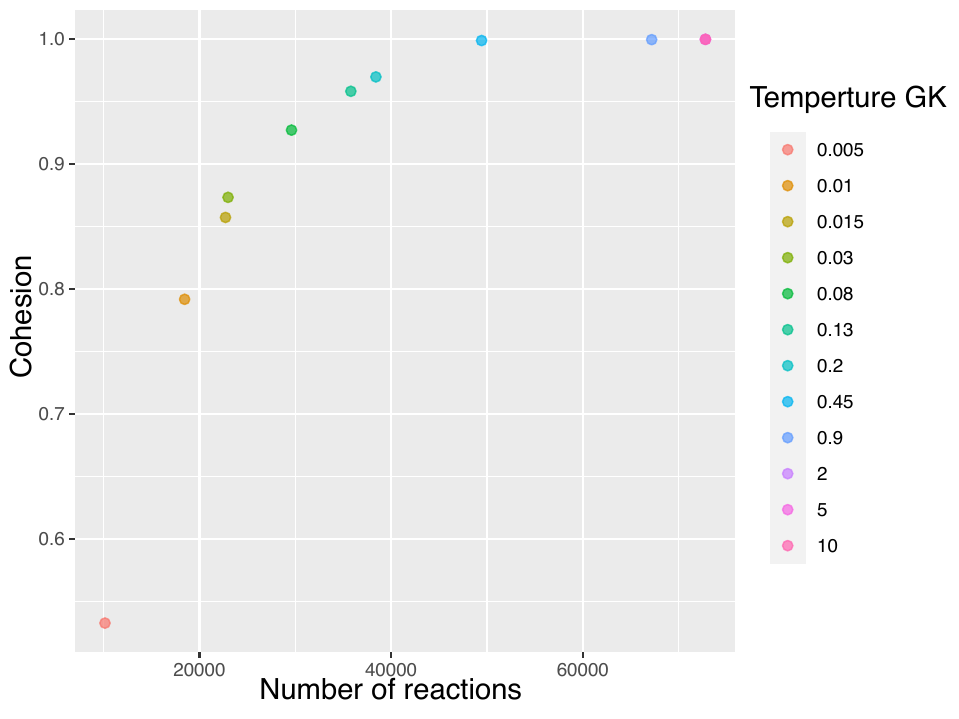}
    \caption{Network cohesion as a function of the number of reactions. The color of the points indicates temperature.}
    \label{fig:nucl_b_coh}
\end{figure}

Thus, as temperature increases and more reactions become activated, the overall cohesion of the network also rises. This trend is clearly illustrated in Figure~\ref{fig:nucl_b_coh}, which shows that once the critical temperature is reached, i.e., $T_c = 0.45$~GK, all networks above this threshold achieve a cohesion value close to one. 

\subsection{Proportion of Semi-Self-Maintaining Atoms}\label{sec:nucl_prop_ssm}

Semi-self-maintaining atoms constitute a small fraction of the total number of atoms; in fact, they represent approximately $3.52\%$ of all atoms. However, this proportion does not remain constant as the temperature changes.

\begin{figure}[h]
    \centering
    \includegraphics[width=0.4\textwidth]{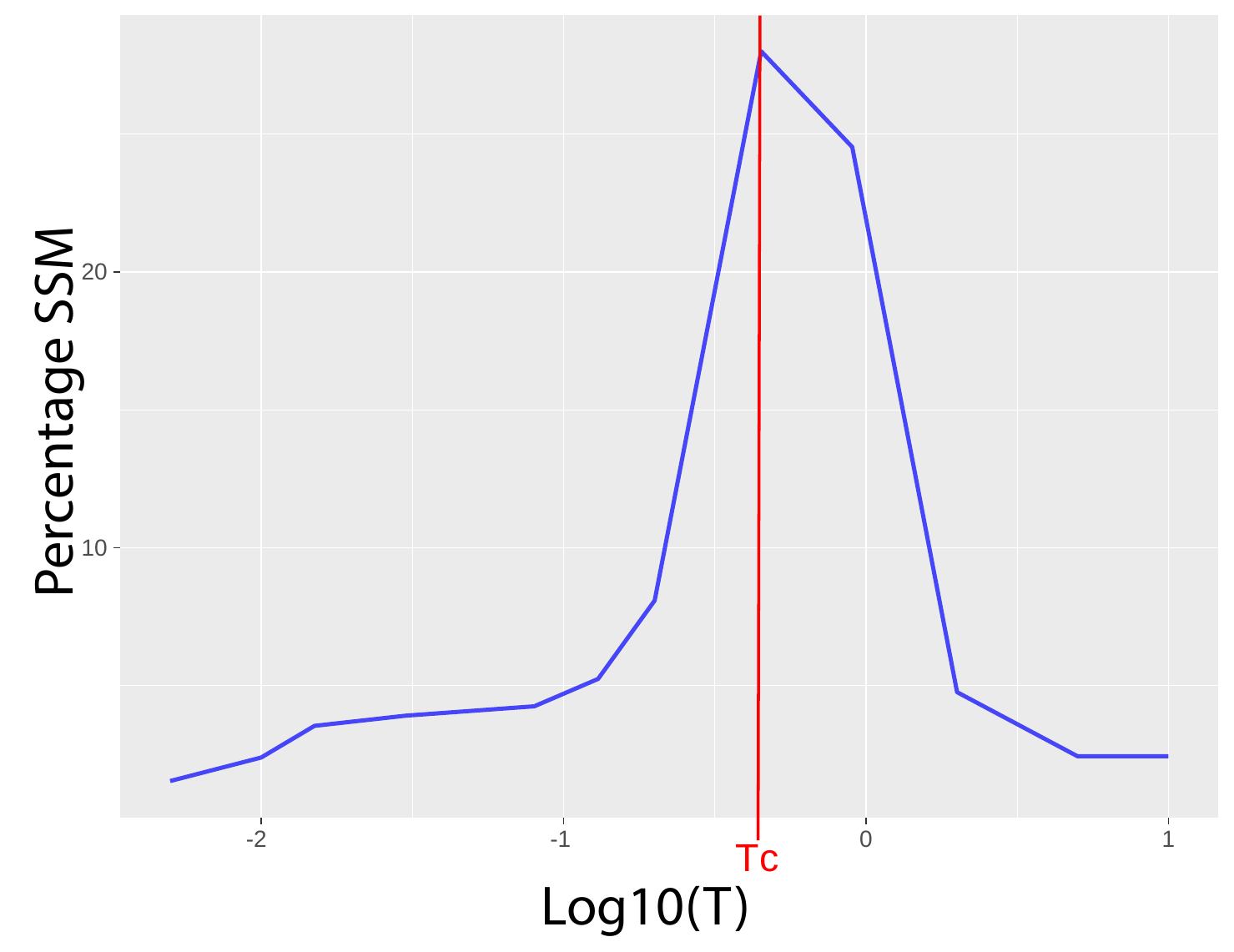}
    \caption{Percentage of semi-self-maintaining (SSM) atoms relative to the total number of atoms for networks at different temperatures.}
    \label{fig:nucl_ssm-ratio}
\end{figure}

Figure~\ref{fig:nucl_ssm-ratio} illustrates a temperature range between approximately 0.45--0.9~GK, within which the proportion of semi-self-maintaining atoms is maximized, reaching around $27\%$. This increase is closely related to the cohesive character of the network. In particular, at $T_c=0.45$~GK, smaller atoms begin to disappear, leading to the initial stages of network cohesion. This structural change improves the probability of identifying potential stable constitutive units. Because the network at these intermediate temperatures is not yet fully cohesive, more atoms of larger size are present. As will be shown later, an increase in the atom size is correlated with a greater probability of being semi-self-maintaining. Beyond 0.9~GK, this property diminishes significantly. At these higher temperatures, the network becomes highly cohesive and consists of only around 40 atoms. Among these, only one atom remains semi-self-maintaining and is contained within all other atoms in the network.

\subsection{Core clusters}\label{sec:nucl_basic_f}

We acknowledge in the previous section that an atom containing other atoms is known as monergy (Lemma~\ref{lem:monergy}). This property is consistently observed in all semi-self-maintaining sets across all temperatures. A equivalence class $P_i$ represents reactions whose closure generates the respective atom $B_i$. A distinction can be made between monergic atoms (with level greater than 1) and basal atoms (with level equal to 1) by analyzing the reactions within their respective equivalence classes. For basal atoms, the number of reactions within the equivalence class is equal to the number of reactions of the atom itself. In contrast, for monergic atoms, the number of reactions in the atom exceeds the number of reactions in the equivalence class that generate it:

\begin{equation}
|\mathcal{P}_i| = \left|\Rs_{\mathcal{B}i}\right| - \left| \bigcup_j \Rs_{\mathcal{B}_j} \right| \quad \text{such that} \quad \mathcal{B}_j \subset \mathcal{B}_i.
\end{equation}

\begin{figure}[h]
\centering
\includegraphics[width=0.5\textwidth]{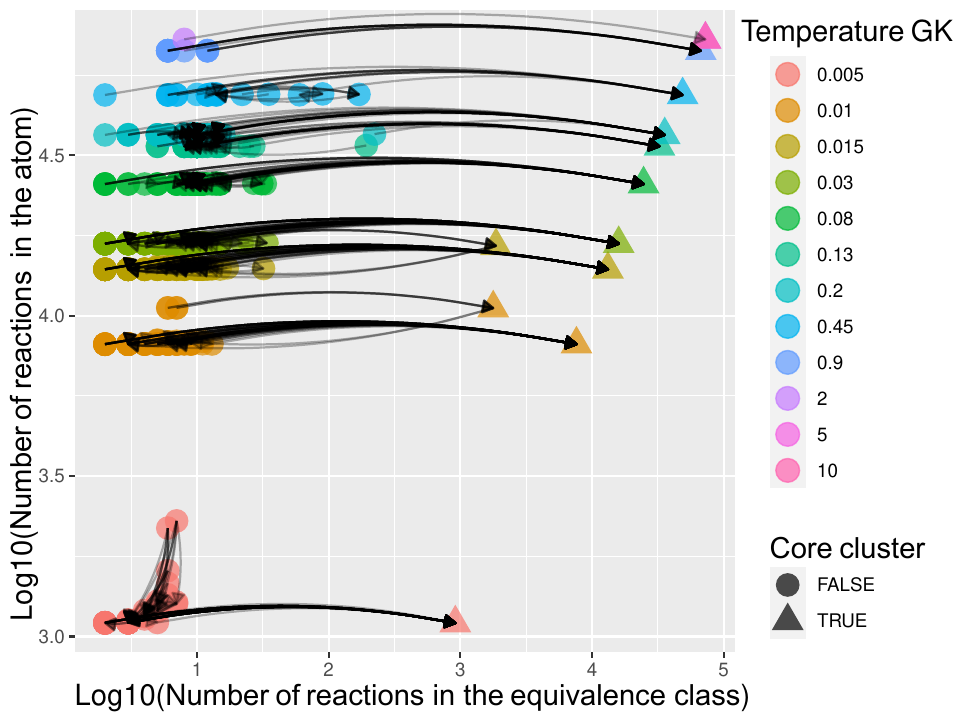}
\caption{Atoms that are semi-self-maintaining and large in size, i.e., containing more than 1100 reactions, are shown. The figure presents the number of reactions in each atom versus the number of reactions in its generating equivalence class, both plotted on a $\mathrm{Log}_{10}$ scale. Core clusters are marked with triangles. Arrows indicate containment relationships, showing which smaller atom is structurally included within each larger one. Notably, every core cluster is contained within all large semi-self-maintaining atoms.}
\label{fig:nucl_bc}
\end{figure}

The possibility of stable nuclear dynamics is exclusively attributed to large semi-self-maintaining atoms. In contrast, small-sized atoms are considered dynamically irrelevant due to their inability to support stable organizational behavior. 
These sets differ in that they involve only reactions of type $(1,1)$ (Table \ref{tab:nuclear-reactions}), consisting of a single support and a single product. Such cases can be represented as directed graphs, where nodes correspond to species, and edges indicate their transformation into others through reactions. In this context, the condition of self-maintenance is satisfied only if all species in the graph belong to closed cycles. This requirement is rarely fulfilled and has been observed in only one case (Equation \ref{eq:nucl_b_ssm}), occurring exclusively at the low temperature of 0.005GK.

\begin{equation}
\begin{gathered}
    r_1: \ce{^{100}Mo} \to \ce{^{100}Te} \\
    r_2: \ce{^{100}Te} \to \ce{^{100}Mo} \\
\end{gathered}
\label{eq:nucl_b_ssm}
\end{equation}

Therefore, the plausibility of dynamic stability is assigned only to large atoms, which contain at least 1100 reactions—differing by approximately three orders of magnitude from the largest among the small-sized atoms. These large semi-self-maintaining atoms vary in the number of reactions present in their generating equivalence classes. As illustrated in Figure~\ref{fig:nucl_bc}, two distinct groups can be identified: those whose equivalence classes contain more than 900 reactions and those with fewer. Atoms with more than 900 reactions are referred to as \textit{core clusters}.

Typically, these atoms are unique per temperature; however, in the temperature range between 0.01 and 0.015GK, two core clusters per temperature are observed (Figure \ref{fig:nucl_bc}). Interestingly, this interval coincides with the formation of two separate clusters, which subsequently merge from 0.030 GK onward. For temperatures equal to or greater than this value, only one core cluster is observed, aligning with the general behavior of all atoms that form a single cluster beyond this temperature. This indicates that core cluster govern the clustering structure of nucleosynthetic units. Each cluster is anchored by a core cluster; every atom within a cluster is a monergy that contains a core cluster. Non-core clusters connect to a core cluster by means of a hierachical chain of atoms (succesive monergies), and the lower bound corresponds precisely to a core cluster. Importantly, all core clusters are semi-self-maintaining and thus potentially stable. This result highlights that the configuration and structure of nucleosynthetic units strongly relate to their inherent stability.

\section{Conclusion}
This paper has extended Chemical Organization Theory (COT) to probe the emergence of complexity and stability within nuclear reaction networks, specifically focusing on stellar nucleosynthesis reactions available in the \mbox{STARLIB} dataset. We introduced and formalized novel structural concepts, including \textit{atoms} as minimal reactive closed sets, alongside \textit{synergy} and \textit{monergy}, to better characterize the mechanisms by which organizations form and persist in these high-energy environments.

Our computational analysis revealed significant temperature-dependent structural dynamics. We demonstrated that network cohesion increases with temperature, and identified an optimal temperature range (approximately 0.45-0.9 GK) where the proportion of semi-self-maintaining atoms is maximized. The central contribution of this work is the identification and characterization of core clusters. These are large, semi-self-maintaining, and inherently monergic structures that exhibit intrinsic stability. Our findings indicate that these core clusters act as core organizational units, structuring other nucleosynthetic assemblages through hierarchical containment. 

We noted that being an organization is a necessary but not a sufficient condition for a subsystem of species to be stable. Therefore, the COT framework aims to identify all possible subsets that could be stable, using a lower computational cost than differential equations approaches such as the one in \citep{wakelam20152014_1}. Instead, by considering all known possible nuclear interactions, we can determine which nuclear synthetic units could plausibly be candidates in terms of stability.

In future work, the evolutionary dynamics of the reaction networks can be studied thanks to a methodology co-developed by one of the authors  \citep{VelozHegeleMaldonado2023}. Starting from a sub-organization and perturbing it through the introduction or removal of species, it becomes possible to study the potential evolutionary trajectories of sub-organizations. To this end, the hierarchical structure of sub-organizations (represented by semi-lattice structure) can be analyzed, and the probabilities of transition among them can be studied. Through an analysis of the transition probabilities across the states of this structure, phenomena such as coexistence and competition can be observed: coexistence when states at the same level present similar transition probabilities between them, and competition when transitions favor one state over another. This approach also highlights the preferred evolutionary pathway toward increasing complexity under perturbations.

Another future avenue of research is to study the possible emergence of nuclear autocatalysis. In the theory of the origin of life developed \citep{williamson2024autocatalytic,xavier2020autocatalytic,peng2022hierarchical} , the emergence of catalysts along the evolutionary pathway plays a crucial role in enabling systems to maintain themselves. We can hypothesize that such models may also apply in the context of nuclear reactions, and not only chemical reactions. 

This study only claims preliminary results that demonstrate the existence of complexity in nuclear reactions. In essence, this study provides a novel framework and initial evidence for understanding how stable, complex structures can emerge from the seemingly chaotic interplay of nuclear reactions. These results offer new perspectives on the principles of self-organization in extreme astrophysical environments and lay foundational groundwork for further exploration into the potential for `Artificial Nuclear Life' and the assembly of complex, high-energy systems.

\footnotesize
\bibliographystyle{apalike}
\bibliography{example} 

\end{document}